\documentclass[11pt]{article}
\normalsize
\usepackage{amsfonts}
\usepackage{amsmath}
\usepackage{amssymb}
\usepackage{latexsym}
\usepackage{enumerate}

\usepackage{oldgerm}
\usepackage{scalerel,stackengine}
\stackMath
\newcommand\reallywidehat[1]{%
\savestack{\tmpbox}{\stretchto{%
  \scaleto{%
    \scalerel*[\widthof{\ensuremath{#1}}]{\kern-.6pt\bigwedge\kern-.6pt}%
    {\rule[-\textheight/2]{1ex}{\textheight}}
  }{\textheight}%
}{0.5ex}}%
\stackon[1pt]{#1}{\tmpbox}%
}
\newtheorem{thm}{Theorem}[section]
\newtheorem{prop}[thm]{Proposition}

\newtheorem{lemma}[thm]{Lemma}
\newtheorem{cor}[thm]{Corollary}

\newtheorem{ex}[thm]{Example}
\newtheorem{examples}[thm]{Examples}
\newtheorem{remark}[thm]{Remark}

\newtheorem{definition}[thm]{Definition}

\newenvironment{proof}{\noindent  Proof:\ }{\hspace*{\fill} $\Box $\\}

\newcommand{\F}{\mathbb{F}}

\newcommand{\ho}{\mbox{\rm Hom}}

\newcommand{\wt}{\mbox{\rm wt}}
\newcommand{\hho}{\mbox{\rm Hom}}

\newcommand{\Sym}{\mbox{\rm S}}

\newcommand{\soc}{\mbox{\rm Soc}}
\newcommand{\ann}{\mbox{\rm Ann}}

\newcommand{\C}{\mbox{$\rm C$}}

\newcommand{\LCD}{\mbox{$\rm  LCD$}}

\newcommand{\cha}{\mbox{\rm char }}
\newcommand{\RM}{\mbox{\rm RM}}

\title{\bf On checkable codes in group algebras }

\author{\textbf{Martino Borello}, \\ Universit\'e Paris 8, Laboratoire de G\'eom\'etrie, Analyse et Applications, LAGA,\\ Universit\'e Sorbonne Paris Nord, CNRS, UMR 7539, F-93430, Villetaneuse, France\\ and \\
\textbf{Javier de la Cruz} \\Universidad del Norte, Barranquilla, Colombia\\
and \\ \textbf{Wolfgang Willems} \\ Otto-von-Guericke Universit\"at, Magdeburg, Germany  \\
 and Universidad del Norte, Barranquilla, Colombia
}
\date{}

\begin{document}

\maketitle


\begin{abstract}   
	We classify, in terms of the structure of the finite group \textit{G}, all
	group algebras \textit{KG} for which all right ideals are right
	annihilators of principal left ideals. This means in the language of
	coding theory that we classify code-checkable group algebras \textit{KG}
	which have been considered so far only for abelian groups \textit{G}.
	Optimality of checkable codes and asymptotic results are discussed.
  \end{abstract}

\section{Introduction}

Block codes were invented in the forties to correct errors in the
communication through noisy channels (see \cite{book} for more
details). In their most general sense, they are just subsets of
(code)words of a fixed length $n$ over an alphabet $K$, such that
the Hamming distance between the words (i.e. the number of distinct
letters) is large enough. One of the main practical problems of
coding theory is how to store a given code, which can be, without an
additional structure, quite expensive (we may have to store the
whole list of codewords). This is one of the main reasons why, since
the beginning, linear codes were introduced: a linear code of length
$n$ over a finite field $K$ is a subspace of the vector space $K^n$.
Such algebraic structure allows to describe a linear code in a more
compact way: a linear code $C$ of length $n$ and dimension $k$ can
be defined by its parity check matrix, which is a $n\times (n-k)$
matrix $H$ such that $c\in C$ if and only if $cH=0$, i.e., it is a
matrix which gives $n-k$ check equations which determine the code.
Such a description reduces exponentially the size of the data to be
stored and it has made linear codes so much used. However, in the
context of McEliece cryptosystem \cite{McEliece} and its dual
version by Niederreiter \cite{Niederreiter}, in which the public key is
given by the parity check matrix, a code of large length and
dimension, the size of the matrix constitutes one of the main
disadvantages. McEliece cryptosystem and its variants, which are
part of the so-called code-based cryptography, are now subjects of
intense research, due to their probable resistance to quantum
computer's attacks. One of the central problems is to reduce the size
of the public key (see for example \cite{BCGO}). In order to do it,
one usually adds more algebraic structure. A classical family of
more structured linear codes is that of cyclic codes, which are
linear codes invariant under a cyclic shift. It is well-known that
they can be seen as (principal) ideals inside the polynomial ring $A
= K[x]/(x^n-1)$. If a linear code $C$ is cyclic, we have $C=gA$ with
$g \in A$ and $C$ is determined by only one check equation, given by
the so-called check polynomial $f= (x^n-1)/g$. In this note we focus
on more general examples in $K$-algebras in which codes are
determined by just one check equation.

A natural generalization of cyclic codes is given by the family of
group codes: a linear code $C$ is called a  {\it $G$-code} (or
a {\it group code}) if $C$ is a right ideal in the group algebra $KG
= \{ a=\sum_{g\in G}a_gg \mid a_g \in G\}$ for  $G$  a finite
group. Here the vector space $KG$ with basis $\{g \in G \}$ serves
as the ambient space with the weight function $\wt(a) = |\{g \in G
\mid a_g \not= 0 \}|$ and the non-degenerate symmetric bilinear form
$\langle \cdot \, , \cdot \rangle $ which is defined by
$$ \langle g,h \rangle = \delta_{g,h} \quad \text{for $g,h \in G$}.$$
Note that $KG$ carries  a $K$-algebra structure
via the multiplication in $G$. More precisely, if
$a=\sum_{g\in G}a_gg$ and $b=\sum_{g  \in G}b_g g$ are given, then
$$ ab = \sum_{g \in G} (\sum_{h \in G} a_hb_{h^{-1}g}) g.$$
In this sense
cyclic codes are  group codes for a cyclic group $G$. Reed Muller
codes over prime fields $\F_p$  are group codes for an elementary
abelian $p$-group $G$ \cite{Berman,Charpin}, and there are many
other remarkable optimal codes which have been detected
as group codes \cite{Bernhardt,Conway,Felde,McLH}.\\

We would like to mention here that choosing right ideals as group
codes is just done by convention.
Everything what we will prove holds equally true for group codes which are left ideals. \\

If $G$ is cyclic, then all right ideals of $KG$ afford only one
check equation as we have seen above. In case $G$ is a
general finite group there are only particular right ideals which
satisfy this property. Such codes are called \emph{checkable} and
these are the subject of this paper. To our knowledge, such codes were
first defined in \cite{JLLX10} and investigated for abelian group
algebras (most of the results of \cite{JLLX10} are now published in
\cite{JLLX}). \\

{In \S\ref{sec:ideals}, we characterize checkable
	codes in terms of their duals, proving that a right ideal $C \leq KG$
	is checkable if and only if its dual is a principal right ideal
	(Theorem \ref{P1}). This result provides an easy way to construct
	random checkable codes and we use it to find optimal codes (see
	Remark \ref{rmk:optimal}). Moreover, two of the consequences are the
	following: maximal ideals in group algebras are checkable and
	two-sided ideals $C$ such that $KG/C$ is a Frobenius algebra are
	checkable. In particular, the Jacobson radical of every group algebra
	is checkable.}

{In \S\ref{sec:algebras} we classify, in terms of
	the structure of the finite group $G$, all group algebras $KG$ for
	which all right ideals are checkable, that is \emph{code-checkable
		group algebras}. This is done in terms of the $p$-blocks of $KG$,
	with $p$  the characteristic of the field $K$ (see Theorem
	\ref{T2}). As a consequence we get the following (see Corollary
	\ref{C3}): $KG$ is a code-checkable group algebra if and only if $G$
	is $p$-nilpotent with a cyclic Sylow $p$-subgroup, which happens if
	and only if all right ideals in the principal $p$-block of $KG$ are
	checkable.}

{In \S\ref{sec:asym} we shortly present the
	asymptotic performance of  checkable codes. Together with the explicit
	construction of optimal codes in Remark \ref{rmk:optimal}, they seem
	to suggest that the family of checkable codes is worth further
	investigation. In particular, it is desirable to prove some
	bounds on the minimum distance for checkable codes and to introduce
	families of checkable codes with prescribed minimum distance (in
	analogy to BCH codes). Some results in this direction for dihedral codes have been given recently in \cite{BJ}. Moreover, it would be extremely interesting
	to develop some fast decoding algorithms. These would be the minimum
	requirements for an effective use of these codes in cryptography
	and this will be the subject of further investigation.}

\section{Checkable ideals}\label{sec:ideals}

Let $A$ be a finite dimensional algebra over a field $K$.
For any subset  $C \subseteq A$, the {\it right annihilator} $\ann_r(C)$ is defined by
$$ \ann_r(C) = \{a \mid a \in A,\ ca =0 \ \text{for all} \ c \in C \}.$$
Analogously, the {\it left annihilator} of $C$ is given by
$$ \ann_l(C) = \{a \mid a \in A,\ ac =0 \ \text{for all} \ c \in C\}.$$
Note that the right (left) annihilators are right (left) ideals in $A$. \\

\begin{definition} {\rm  A right ideal $ I\leq A$ is called {\it checkable} if there exists an element $v \in A$ such that
		$$ I=  \{a \mid a \in A,\, va=0\} = \ann_r(v) = \ann_r(Av).$$
		Note that checkable left ideals are defined analogously via the left
		annihilator of a principal right ideal. A group algebra $KG$ is
		called {\it code-checkable} if all right ideals of $KG$ are
		checkable. }
\end{definition}

Recall that a finite dimensional $K$-algebra $A$ is called a {\it Frobenius algebra} if there exists
a $K$-linear function $\lambda \in  \ho_K(A,K)$ whose kernel contains no left or right ideal other than zero.
In case $\lambda(ab) = \lambda(ba)$ for all $a,b \in A$, we say that $A$ is a {\it symmetric algebra}.
Note that group algebras are Frobenius algebras;  even more, they are symmetric
algebras. In such algebras the annihilators of ideals satisfy the
double annihilator property  (see \cite[Chap. VII]{HB}).

\begin{prop}[Double Annihilator Property] \label{P4}    Let $A$ be a Frobenius algebra.
	If $I \leq A$ is a right ideal in $A$, then
	$$ I = \ann_r(\ann_l(I)).$$
	A similar equation holds for left ideals.
\end{prop}

\begin{cor} \label{cor1} In a Frobenius algebra $A$ a right $($left$)$ ideal $I$ is checkable if and only if $\ann_l(I)$ $(\ann_r(I))$ is a
	principal left $($right$)$ ideal.
\end{cor}

\begin{examples}{\rm  a) Let $e=e^2$ be an idempotent in $A$. Then the ideal $eA$ is checkable. This can be seen as follows. Obviously,
		$eA \leq \ann_r(A(1-e))$. Since any $0 \not= (1-e)b \in (1-e)A$ is not in $\ann_r(A(1-e))$ we have
		$eA = \ann_r(A(1-e))$.\\
		b) If $A$ is a semisimple algebra, then all right and left ideals are generated by idempotents. Thus all right and left ideals are checkable.\\
		c) All cyclic codes are checkable, since the check equation is given by the check polynomial.\\
		d) $\LCD$ group codes $C$ (that is, codes for which $C\cap C^\perp
		=\{0\}$) are checkable since $C=eKG$ with a self-adjoint idempotent
		$e$, by \cite{CW}.}
\end{examples}

As mentioned in the introduction, the group algebra $KG$ carries a non-degenerate symmetric bilinear form $\langle \cdot \, , \cdot \rangle $.
Thus, for any subset $C \subseteq KG$ the orthogonal space $C^\perp \leq KG$ is well defined. Observe that $C^\perp$ is always a $K$-linear
vector space.

In order to state an early result of Jessie  MacWilliams  recall
that the $K$-linear map  \ $\hat{}: KG \longrightarrow KG$ defined
by $g \mapsto \hat{g}=g^{-1}$ ($g \in G$) is an antialgebra
automorphism of $KG$.

\begin{lemma}[\cite{MacW}] \label{dual}  If $C$ is a right ideal in $KG$, then $\C^\perp = \widehat{\ann_l(C)}$.
	Similarly, for a left ideal $C$ we have $\C^\perp =
	\widehat{\ann_r(C)}.$
\end{lemma}
\begin{proof} We have $ a = \sum_{g \in G}a_gg \in \widehat{\ann_l(C)}$ if and only if $\hat{a}ch =0$, for all $c \in C$ and all $h \in G$.
	Since the coefficient at $h$ in $\hat{a}ch$ equals
	$$ \sum_{g \in G} a_gc_g = \langle a,c \rangle, $$
	the assertion follows.
\end{proof}

\begin{thm}\label{P1} For any right ideal $C \leq KG$ the following are equivalent.
	\begin{itemize}
		\item[\rm a)] $C$ is checkable.
		\item[\rm b)] $\C^\perp$ is a principal right ideal.
	\end{itemize}
\end{thm}
\begin{proof}
	According to Corollary \ref{cor1}, $C$ is checkable if and only if $\ann_l(C) = KGv$, for some $v \in KG$.
	Now Lemma \ref{dual} implies $C^\perp = \reallywidehat{\ann_l(C)} = \reallywidehat{KGv} = \hat{v}KG$ and the proof is complete.
\end{proof}

\begin{examples} {\rm a) The binary extended Golay ${\cal G}_{24}$ is a  group code in $\F_2\Sym_4$, with $\Sym_4$  the symmetric 
		group on $4$ letters
		(see \cite{Bernhardt}) and a group code in $\F_2{\rm D}_{24}$, with ${\rm D}_{24}$  a dihedral group of order $24$ (see \cite{McLH}).
		Note that the binary extended Golay code is checkable in both algebras, by Theorem \ref{P1}, since ${\cal G}_{24} ={\cal G}_{24}^\perp $
		is constructed as a principal ideal in both cases.
		We would like to mention here that
		the extended ternary Golay code is not a group code according to (\cite[Theorem 1.1]{W}).\\
		b) In \cite{Berman} and \cite{Charpin} it is shown that the
		Reed-Muller code $\RM_p(r,m)$ of order $r$ and length $p^m$ over the
		prime field $\F_p$ can be constructed as $\RM_p(r,m) = J^{N-r}$,
		with $J$  the Jacobson radical of a group algebra $\F_pG$, for
		$G$  an elementary abelian $p$-group of
		rank $m$ and $N=m(p-1)$. For more details in what follows the reader may inquire \cite{W99}. Note that 
		$(J^{N-r})^\perp =\RM_p(r,m)^\perp = \RM_p(N-r-1,m) =J^{r+1}$. Thus
		in order to check which Reed-Muller codes are checkable, we have to
		check which powers $J^i$ are principal ideals. Clearly, if $m=1$,
		then $G$ is cyclic and therefore all ideals in $\F_pG$ are
		principal. In case $m>1$, apart from the full space $\F_pG$, only
		$\RM_p(N-1,m) = J$ is checkable. This can be seen as follows.
		Clearly, $J^\perp = J^N = \RM_p(0,m) = \F_p\sum_{g \in G} g$ is
		principal. Suppose that $J^r$ is principal, for some $0<r<N$. Thus
		$J^r/J^{r+1}$ is principal as well, hence $J^r/J^{r+1} =  (a
		+J^{r+1})\F_p G$. Moreover, $J^r/J^{r+1}$ is a direct sum of trivial
		$\F_pG$-modules. Thus $(a+ J^{r+1})g = a + J^{r+1}$, for all $g \in
		G$, which implies $\dim J^r/J^{r+1} =1$. On the other hand,
		according to (\cite[Proposition 7.2.3]{W99}) we have $\dim
		J^r/J^{r+1} >1$. Thus only $J^N$ is principal, which means, by
		Theorem \ref{P1}, that $J=\RM_p(N-1,m)$ is the only checkable
		Reed-Muller code apart from $\F_pG$. }
\end{examples}

\begin{remark}{\rm In  \cite{JLLX10} the authors point out that in numerous cases the parameters of checkable group
		codes for an abelian group $G$ are as good as the best known linear
		codes mentioned in \cite{G}. Even more, there is a checkable
		$[36,28,6]$ group code in $\F_5(C_6 \times C_6)$ and a checkable
		$[72,62,6]$ group code in $\F_5(C_6\times C_{12})$. In both cases
		the minimum distance is improved by $1$ from an earlier lower bound
		in \cite{G}.}
\end{remark}

\begin{remark}\label{rmk:optimal}{\rm Theorem \ref{P1} provides an easy way to construct random
		checkable codes: it is enough to choose a random element in $KG$ and
		then take the dual of the principal ideal generated by this element.
		Such a construction allows to do extensive searches for codes with
		the best known minimum distance (let us call them optimal codes, for
		simplicity). Using {\sc Magma} \cite{Magma}, we observed that there
		exists an optimal checkable code over $\F_2$, for every group of
		order $\leq 100$, and an optimal checkable code over $\F_3$ and
		$\F_4$, for every group of order $\leq 50$. For some groups we could
		find only trivial checkable codes, that is, of dimension $1$. But in
		many cases the optimal checkable codes that we found have a higher
		dimension.\\
		Let us give two examples: a binary and a ternary optimal code.\\
		If
		$$\begin{array}{ll}G=&\langle
		a,b,c,d \mid\\ a^8=&b^2=c^2=d^8=1,ab=ba,ac=ca,bc=cb,da=a^7cd,db=bd^5,dc=cd\rangle\end{array}$$
		(this is the 22nd group of order $64$ in the library SmallGroups of
		{\sc Magma}) and
		$$u=1 + a^6c + ad^4 + a^3 + a^7bd^4 + a^7cd^4 + a^7bc + a^7bcd^4
		+ d + a^6d + acbd + a^7d^5\in \F_2 G,$$ the dual of $uKG$ is a
		$[64,32,12]$ code over $\F_2$.\\
		If $$G=\langle
		a,b,c \mid a^4=b^4=c^3=1,ab=ba,ca=a^3b^3c,cb=ac\rangle\simeq (C_4\times
		C_4)\rtimes C_3$$ and
		$$v=1+2b+a^3b^2+2a^3+2a^3b^3+2c^2b^3+c^2ab^3\in \F_3G,$$
		the dual of $vKG$ is a $[48,15,18]$ code over $\F_3$.\\
		Note that we can describe these codes very easily in terms of a
		check element, which is a generator of the dual. Here we chose a
		check element of minimum weight. }
\end{remark}

Recall that for a right $KG$-module $M$ the dual $KG$-module $M^*$ is defined by the $K$-vector space $M^* = \hho_K(M,K)$ on which $G$ acts from the right by
$$   (\varphi g)(m) = \varphi(mg^{-1}),  \quad \text{for} \ \varphi \in M^*, g \in G, m \in M.$$

\begin{remark} {\rm If $C$ is checkable, then $C^* \cong KG/uKG$, for some $u \in KG$. This immediately follows from Theorem \ref{P1}
		applying $KG/C^\perp \cong C^*$ which has been proved in
		(\cite[Proposition 2.3]{W}). }
\end{remark}

\begin{cor} Maximal ideals in $KG$ are checkable.
\end{cor}
\begin{proof}  For a $KG$-module $M$, let $l(M)$ denote the composition length of $M$, i.e.,
	the number of irreducible composition factors in a Jordan-H\"older series of $M$. Let $l(KG) =l$. Thus, if
	$C$ is a maximal ideal in $KG$, then $l(C) = l(C^*)=l-1$. Since
	$KG/C^\perp \cong C^*$ we get  $l(C^\perp)=1$.  Hence $C^\perp$ is a minimal ideal in $KG$.
	But minimal ideals in $KG$ are principal. Hence the assertion follows by
	Theorem \ref{P1}.
\end{proof}

Minimal ideals are in general not checkable as the following result shows.

\begin{prop} Let $G$ be a finite $p$-group and  $\cha K = p$. Then the minimal ideal  $C=K\sum_{g\in G}g$ in $KG$
	is checkable if and only if $G$ is cyclic.
\end{prop}
\begin{proof} Let $g_1,\ldots,g_s$ be a minimal set of generators of $G$. Then
	$C^\perp =
	J(KG)$, where $J(KG)$ is the Jacobson radical of $KG$. One easily
	sees that $J(KG) = \sum_{i=1}^s (g_i -1)K$ is principal if and only if $s=1$, i.e., $G$ is cyclic. 
	Now the assertion follows, by Theorem \ref{P1}.
\end{proof}

\begin{prop} Let $C$ be a two-sided ideal in $KG$ such that $KG/C$ is a Frobenius algebra. Then $C$ is checkable.
	In particular, in this case $J(KG)$ is checkable.
\end{prop}
\begin{proof} Since
	$KG/C$ is a Frobenius algebra, a result of Nakayama (\cite[Theorem
	A]{T}) directly implies  $\ann_l(C) =KG a$, for some $a \in A$.
	Applying the double annihilator property, we get  $C = \ann_r(KGa)$.
	Finally note that according to Wedderburn's Theorem $KG/J(KG)$ is a
	direct sum of full matrix algebras over extension fields of $K$,
	hence  a Frobenius algebra.
\end{proof}

\section{Code-checkable group algebras}\label{sec:algebras}

Let $KG = B_0 \oplus \ldots \oplus B_s$ be a decomposition of $KG$
into $p$-blocks $B_i$, with $\cha K=p$ and $B_i=f_iKG$ with block
idempotents $f_i$. Recall that the $B_i$ are $2$-sided ideals and as
such, indecomposable. They  are uniquely determined by $KG$.
Furthermore, the $f_i$ are primitive idempotents in the center of
$KG$. For more details the reader is referred to (\cite[Chap. VII,
Section 12]{HB}).

If $C \leq KG$ is a group code, then $C = Cf_0 \oplus \ldots \oplus
Cf_s$. One easily sees that $C$ is checkable, i.e. $C =\ann_r(KGa)$,
if and only if
$$Cf_i = \ann_r^{B_i}(KGf_ia) = \ann_r^{B_i}(B_i f_ia) = \{ b \in B_i \mid (f_ia) b=0\}.$$
This shows that $C$ is checkable if and only if the block components $Cf_i$ are checkable in $B_i$, for all $i$. \\

For an algebra $A$ a right (left) $A$-module $M$ is called {\it
	uniserial} if it has only one Jordan-H\"older series, or in other
words $M$ has only one composition series. Furthermore, the {\it
	projective cover} $P(M)$ of an irreducible $A$-module $M$ is an
indecomposable projective $A$-module which has a factor module
isomorphic to $M$ (see (\cite[Chap.VII]{HB})).

\begin{thm} \label{T2} Let $\cha K =p$ and let $B$ be a $p$-block of $KG$. Then the following are equivalent.
	\begin{itemize}
		\item[\rm a)] All right ideals in $B$ are checkable.
		\item[\rm b)] All left ideals in $B$ are principal.
		\item[\rm c)] $B$ contains only one irreducible left module $M$ whose projective cover $P(M)$ is uni\-serial.
		\item[\rm d)] The defect group of $B$ is cyclic and $B$ contains only one irreducible left module.
	\end{itemize}
\end{thm}
\begin{proof} First note that $B$ is a symmetric algebra (\cite[Chap.VII, Section 11]{HB}). \\
	$a) \Longleftrightarrow b)$  Let $I$ be a right ideal of $B$. Then
	$I= \ann_r^B(Bv)$, for some $v \in B$, if and only if $\ann_l^B(I) =
	Bv$. Since $\ann_r^B$ yields a bijection from the set of left ideals
	in $B$ onto the set of right ideals in $B$
	we are done. \\
	$b) \Longrightarrow c)$ Clearly, all left ideals are principal if
	and only if all right ideals are principal, just by applying the
	antiautomorphism \ $\widehat{}$ \ . Thus $B$ is an artinian principal
	ideal ring and (\cite[Theorem 2.1]{EG}) implies that all left
	$B$-modules are homogeneous uniserial. In particular all projective
	indecomposable left $B$-modules are uniserial.
	
	Next we have to show that $B$ contains only one irreducible left module.
	Let $J=J(B)$ denote the Jacobson radical of $B$ and  $M$ be an irreducible left
	$B$-module with projective cover $P=P(M)$.
	Let $X$ be the largest submodule of $\soc_l(P/J^k P)$ whose irreducible components are all isomorphic to $M$.
	Since the full preimage
	of this completely irreducible module is a left ideal in $B$, it is principal, hence a factor module of $B$.  This implies that
	$X$ is a factor module of $B/JB$, for all $k$. Doing this argument for a decomposition of $B$ into a direct sum of projective
	indecomposable left modules and counting all composition factors in the regular left module $B$, we see  that all
	composition factors of $P(M)$ are isomorphic to $M$. \\
	Since the principal indecomposable modules in a block are connected, we get that $B$ contains only one irreducible left module.\\
	$c) \Longrightarrow b)$  Let $I$ be a left ideal of $B$. If there exists $a \in B$ such that $I=Ba +JI$, then by Nakayama's Lemma $I =Ba$
	and we are done. Thus it is sufficient to show that $I/JI$ is generated by one element as a $B$-left module.\\
	First note that
	the condition in c) obviously implies that there is also exactly one irreducible right module in $B$ whose projective cover is uniserial.
	The corresponding irreducible right module is $\hat{M}=\ho_K(M,K)$ with the right structure
	$$   (\varphi g)(m) = \varphi(gm)  \quad \text{for} \ \varphi \in \hat{M}, g \in G, m \in M.$$
	Thus a result of Nakayama (\cite[Theorem 17]{N}) says that all left
	(and right) $B$-modules are uniserial. Clearly, $I/JI$ is  a
	$B/J$-submodule of $\soc(B/JI)$. Since $B/JI$ is uniserial, we see
	that $I/JI$ is a submodule of $B/J$. But all ideals in $B/J$ are
	principal. Hence $I/JI$ is generated by one element. \\
	$c) \Longrightarrow d)$. Note that b) is exactly the statement in
	(1) of (\cite[Theorem 2.1]{EG}) which is equivalent to (3). Thus all
	indecomposable left $B$-modules are submodules of $B$ which means
	that $B$ is of finite representation type. Now d) follows by (\cite{Z}, Proposition 2.12.9).\\
	$d) \Longrightarrow c)$ Again by (\cite{Z}, Proposition 2.12.9), we
	know that $B$ is of finite representation type. Following the proof
	of (\cite{A}, Section 18, Proposition 3), we see that the projective
	cover $P(M)$ of the unique irreducible left module $M$ in $B$ is
	uniserial.
\end{proof}

As usual the {\it principal $p$-block} $B_0(G)$ of $G$ is the block which contains the trivial $KG$-module.
Furthermore, remember from finite group theory that a group $G$ is called {\it $p$-nilpotent} if $G$ has a normal $p'$-subgroup $N = O_{p'}(G)$
such that the factor group $G/N$ is a $p$-group, i.e.,
$G/N$ is isomorphic to a Sylow $p$-subgroup of $G$. \\

\begin{cor} \label{C3} Let $\cha K =p$ and let $B_0(G)$ be the principal $p$-block of $KG$. The following are equivalent.
	\begin{itemize}
		\item[\rm a)] $G$ is $p$-nilpotent with
		a cyclic Sylow $p$-subgroup.
		\item[\rm b)]  $KG$ is a code-checkable group algebra.
		\item[\rm c)] All right ideals in $B_0(G)$ are checkable.
	\end{itemize}
\end{cor}
\begin{proof} $a) \Longrightarrow b)$ By (\cite[Chap. VII, Theorem 14.9]{HB}), each block has exactly one irreducible
	left module. Since the principal block $B_0(G)$, which contains the trivial module $1_G$, is isomorphic to
	$KG/O_{p'} (G) \cong KT$, with $T$ a Sylow $p$-subgroup of $G$, we get that $B_0(G)=P(1_G) \cong KT$ is uniserial. 
	If $M$ is any irreducible $KG$-module, then $P(M)$ is a factor module of $P(1_G) \otimes M$, which is uniserial.
	Thus $P(M)$ is uniserial as well. This shows that all blocks of $KG$ satisfy condition c) of Theorem \ref{T2}, hence condition a),
	which means that $KG$ is code-checkable. \\
	$ b) \Longrightarrow c)$ This is obvious.\\
	$c) \Longrightarrow a).$  Note that according to Theorem \ref{T2} the condition in c)
	implies that $B_0(G)$ contains only one irreducible left module, namely the trivial module $1_G$, and $P(1_G)$ is uniserial.
	Thus again by (\cite[Chap. VII, Theorem 14.9]{HB}), the group $G$ must be
	$p$-nilpotent. In particular, if $T$ is a Sylow $p$-subgroup of $G$,
	then $B_0(G) \cong KG/O_{p'} (G) \cong KT \cong P(1_G)$ forces $T$ to be cyclic.
\end{proof}

Observe that the equivalence of a) and b) is already contained in an
early paper of  Passman (\cite[Theorem 4.1]{P}).

\begin{remark} {\rm In \cite{JLLX} the authors study group codes in code-checkable group algebras $KG$, for $G$ abelian, i.e.,
		$G=A \times T$, with $A$  an abelian $p'$-group and $T$  a cyclic $p$-group, if $p = \cha K$. In particular,
		a characterization and enumeration of Euclidean self-dual and  self-orthogonal group codes is given.
		
	}
\end{remark}

\section{Asymptotically good classes}\label{sec:asym}

In the literature there are many papers which prove that particular classes of codes are asymptotically good.
In \cite{BM} the authors investigated binary group codes over
dihedral groups of order $2m$, for $m$  odd. Their results in
Section 4  show that the class of group codes over these groups is
asymptotically good. Applying field extensions as in (\cite{FW},
Proposition 12) this result can be extended to any field of
characteristic $2$. These methods have been generalized to any
finite field in odd characteristic \cite{BW}.
Thus we have the following result.

\begin{thm}[\cite{BM}, \cite{BW}]  For any finite field $K$  the class of group
	codes in code-checkable group algebras is asymptotically good.
\end{thm}

\end{document}